\documentclass[10pt, conference, letterpaper]{IEEEtran}
\IEEEoverridecommandlockouts
 \usepackage[colorlinks,
             linkcolor=blue,
             anchorcolor=blue,
             citecolor=blue,
             bookmarks=false]{hyperref}
\hypersetup{bookmarks=false}
\usepackage{textcomp}
\usepackage{xcolor}
\usepackage{bbm}
\usepackage{listings}
\usepackage{cite}
\usepackage{url}
\usepackage{xcolor}
\usepackage{graphicx}
\usepackage{epstopdf}
\usepackage[cmex10]{amsmath}
\usepackage{cases}
\usepackage{amsfonts}
\usepackage{amssymb}
\usepackage{amsthm}
\usepackage{mathrsfs}
\usepackage{verbatim}
\usepackage{algorithm}
\usepackage{algpseudocode}
\usepackage{setspace}
\usepackage{bm}%
\usepackage{stfloats}
\usepackage{enumerate}

\theoremstyle{plain}
\newtheorem{theorem}{Theorem}
\newtheorem{remark}{Remark}
\newtheorem{lemma}[]{Lemma}
\newtheorem{corollary}{Corollary}

\usepackage{subfig}
\usepackage{multirow} 
\usepackage{dsfont}
\begin{document}
% \title{Stay Fresh by Probing: AoI-Optimized Random Access with Energy Harvesting}
% \title{Probe-then-Send: Enhancing Information Freshness in Energy-Harvesting Random Access Networks}
%\title{Probe to Stay Fresh: Enhancing Age of Information in Energy-Harvesting Random Access Networks}
\title{Probing for Better Age of Information in Energy-Harvesting Random Access Networks}

\author{
  \IEEEauthorblockN{Ziyi~Li, 
         Fangming~Zhao, 
         and Howard H.~Yang}
  \IEEEauthorblockA{%
                    Zhejiang University/University of Illinois Institute, Zhejiang University, Haining, China\\
    ziyi2.25@intl.zju.edu.cn, 
    fangming.23@intl.zju.edu.cn,
    haoyang@intl.zju.edu.cn}
}
% \title{Channel Probing for Fresh Information in Energy-Harvesting Random Access Networks}

% \title{Probing for Better Age of Information in Energy-Harvesting Random Access Networks}

%\author{Ziyi~Li, Fangming~Zhao, and Howard H.~Yang}

\maketitle

\begin{abstract}
In this paper, we investigate the impact of channel probing and reservation on the Age of Information (AoI) in energy-harvesting (EH) random access networks, where each source relies solely on harvested energy for status updating. 
To mitigate collisions, each node may expend a small amount of energy to send a probing signal before transmission, and a successful probe reserves the channel in the current slot. 
If probing fails, the node can either remain silent, termed strict avoid free competition (SAFC), attempt data transmission with a certain probability, termed reserved nodes competition (RUC), or adopt all-active nodes competition (AUC), where all energy-sufficient nodes may contend regardless of whether they probed.
We derive closed-form expressions for the network-average AoI under these three schemes and validate them via simulations. The results show that AUC consistently achieves the lowest AoI by shortening the waiting time to convert harvested energy into successful updates. This finding challenges the conventional wisdom that strict collision avoidance is always optimal in energy-constrained systems, since allowing additional contention can effectively amortize probing overhead across more transmission opportunities. Comparisons with EH-enabled slotted ALOHA further show that probing-based access significantly outperforms direct transmission in energy-constrained regimes, highlighting channel probing as an effective approach to improving freshness. 
\end{abstract}

\section{Introduction}
Equipping transmitters in random access networks with energy harvesting (EH) capabilities provides a natural pathway toward sustainable Internet of Things (IoT) systems \cite{famitafreshi2021comprehensive}. 
In such systems, timeliness is as important as reliability, making the Age of Information (AoI) \cite{yangUnderstanding} a key performance metric. 
Unlike battery-powered networks, EH-enabled nodes suffer from stochastic energy arrivals, resulting in unavoidable transmission interruptions for energy accumulation. 
Such intermittency fundamentally limits update regularity and degrades AoI performance \cite{TITAhmedAoIEH}. 

The AoI concept was originally introduced to quantify data freshness in real-time status update systems, with early EH studies mainly focusing on point-to-point policies \cite{RoyISIT2015,2018ISIT,ISIT2018finiteEH,WuXianwenTGCN}. 
As research extended to multi-source networks, the impact of inter-link interference and battery-aware access adaptation on AoI has been highlighted in \cite{ChenNikosINFOCOM,EHAoI:TCOM2025,zhao2025age}. 
Several works further optimized slotted ALOHA (SA) from an AoI perspective by balancing transmission success probability and update frequency \cite{sakakibara2018modeling,EHAoI:TCOM2025,zhao2025age,xiao2024information}. 
Nevertheless, SA still suffers from collisions, which are particularly harmful in EH systems because failed transmissions waste scarce harvested energy and prolong recharge periods.

To mitigate collision costs, probing-based mechanisms have been widely adopted. 
For example, \cite{wang2022age} showed that probing and channel reservation can outperform ALOHA in terms of AoI under high traffic loads, while related MAC designs improve access stability via reservation and backlog control \cite{crowther1973system,thomopoulos2002simple,chung1992diversity,alsbou2015analysis,vazquez2015reservation}. Motivated by these insights, we integrate probing into channel access, where reservation failure can be handled from aggressive fallback contention to strict silence, enabling us to quantify the trade-off among collision avoidance, probing overhead, and recharge delay. 
This perspective is particularly relevant in EH systems, where fewer collisions do not necessarily imply lower AoI.

The contributions of this paper are outlined as follows:
\begin{itemize}
    \item We develop a theoretical framework that accounts for the interplay among channel probing and reservation, energy harvesting and consumption, and interference from concurrent transmissions in the network AoI analysis. 
    \item We propose three channel access mechanisms, termed all-active nodes competition (AUC), reserved nodes competition (RUC), and strict avoid free competition (SAFC), ordered from the most aggressive to most conservative. 
    \item We validate the accuracy of the analysis through simulations, where the results demonstrate that the AUC strategy achieves the optimal network AoI performance and that reservation-based access is advantageous in energy-constrained regimes.
\end{itemize}

\section{System model}

\subsection{Network Configuration}
We consider a random-access network comprising $n$ source nodes and a common access point (AP), in which the source nodes can harvest energy from the ambient environment and use it to update status information toward the AP over a shared spectrum. 
% We segment the time into equal-length intervals indexed by $t \in \mathbb{N}$, referred to as communication rounds. 
Time is divided into communication rounds indexed by $t \in \mathbb{N}$.
% We adopt the collision model \cite{CollisionModel} to characterize data transmission, namely, in each communication round, transmissions fail if more than one node accesses the spectrum. 
Under the collision model \cite{xiao2024information}, transmissions fail if more than one node accesses the spectrum.

We model the energy arrivals at each source node as an independent Bernoulli process with rate $\xi$, i.e., a typical source node harvests a single unit of energy with probability $\xi$ at each communication round.
The source node stores each incoming energy unit into an energy buffer. 
To facilitate analysis, we assume this buffer has infinite capacity.
In this network, we assume that each source node requires $M$ units of energy to transmit a single packet.
To avoid collisions among transmitters, we introduce a reservation-based mechanism prior to each transmission attempt.
Specifically, we assume that each node uses one energy unit to send a probing signal, and that transmitting a status update packet consumes $M$ energy units. 
Therefore, a source node remains silent until its energy buffer accumulates to at least $M+1$ units of energy.

% Consequently, we further structure each communication round into a probing phase followed by a data transmission phase. 
Each round consists of a probing phase followed by data transmission.
Regarding the physical time scales, we denote the duration of the data transmission phase as $L$, which is normalized to $1$ for analytical convenience. 
The duration of each probing phase is $\delta L$, where $\delta$ represents the probing-to-data slot ratio. 
As such, the physical duration of a single communication round consisting of one mini-slot is given by $T_{\mathrm{round}} = (1 + \delta)L$. 

\subsection{Transmission Protocols}
At a typical source node, when the harvested energy is sufficient for a single status-update operation (comprising channel probing and data transmission), it will initiate a transmission attempt. 
% Specifically, the source node sends a probing signal with probability $q$, consuming one unit of energy. 
Each node probes with probability $q$, consuming one energy unit.
The probing attempt succeeds when the AP receives only one node signal, after which the radio channel is dedicated to the node's data transmission in the subsequent mini-slot; otherwise, the probing fails. 

Upon receiving a failed channel reservation, the source node can choose to stay in silent until the confirmation of the next probing attempt, or directly access the channel, at the risk of incurring transmission collisions.
In this paper, we investigate the following three mechanisms, if a source node \textit{fails} the channel reservation: 
\begin{itemize}
    \item \textit{All-active Nodes Competition (AUC)}: All the nodes with harvested energy sufficient to support (at least) one transmission can update in the data transmission phase, independently with probability $\eta$, even if the node has not performed channel probing.
    \item \textit{Reserved Nodes Competition (RUC)}: Only the nodes that have performed channel probing can update in the data transmission phase, independently with probability $\eta$. 
    \item \textit{Strict Avoid Free Competition (SAFC)}: All the nodes are prohibited from data transmission if the channel reservation fails. Consequently, data transmission occurs if and only if the channel is successfully reserved by a certain source node during the probing phase.
\end{itemize}

\subsection{Performance Metric}
Let $\Delta_j(t)$ denote the instantaneous AoI of source node $j$ at the beginning of round $t$. Formally, the AoI evolves as 
\begin{equation}
    \Delta_j(t)=t-G_j(t),
\end{equation}
where $G_j(t)$ is the generation time of the latest packet received over this link at time $t$. 
In this paper, we focus on the network-average AoI $\bar{\Delta}$ \cite{yangUnderstanding}, defined as
\begin{equation}
    \bar{\Delta}=\frac{1}{n}\sum_{j=1}^{n}\lim_{T\to \infty} \frac{1}{T}  \sum_{t=1}^{T}\Delta_j(t).
\end{equation}

% \begin{figure*}
%     \centering
%     \includegraphics[width=0.8\linewidth]{figure/Markov_new.drawio.pdf}
%     \caption{A discrete-time Markov chain to characterize the energy harvesting and consumption process.}
%     \label{fig:markov}
%     %\vspace{-0.4cm}
% \end{figure*}

% ======================= %
%       Analysis
% ======================= %

\section{Analysis}

\subsection{Unified Analysis}
We develop a unified analytical framework for studying the network-average AoI across the three mechanisms. Specifically, we denote by $T$ and $p_\mathrm{s}$ the interval between any two consecutive status updates and the transmission success probability, respectively. 
Then, following \cite{zhao2025age}, the network-average AoI is given by  
\begin{equation} \label{eq:aoi}
    \begin{split}
        \bar{\Delta} 
        = \frac{\mathbb{E}[T^2]}{2\mathbb{E}[T]} + \left(\frac{1}{p_\mathrm{s}}-1\right)\mathbb{E}[T] + \frac{1}{2}.
    \end{split}
\end{equation} 
Thus, the analysis reduces to $p_\mathrm{s}$, $\mathbb{E}[T]$, and $\mathbb{E}[T^2]$.

Let $p_\mathrm{a}$ denote the probability that a node has at least $M+1$ energy units. Moreover, let $p_{\mathrm{c},\mathrm{s}}$ be the transmission success probability in the channel contention phase and let $p_{\mathrm{ac}}(s)$ be the data transmission mini-slot access probability (given a channel reservation failure), where $s \in \{ \mathrm{AUC}, \mathrm{RUC}, \mathrm{SAFC}\}$.
According to the mechanism principles, we have 
\begin{equation} \label{eq:p_access_cases}
   p_{\mathrm{ac}}(s) =
       \begin{cases}
           \eta, & s = \mathrm{AUC}, \\
           q\eta, & s = \mathrm{RUC}, \\
           0, & s = \mathrm{SAFC}.
       \end{cases}
\end{equation}
As such, the transmission success probability of a typical source node is given by the following. 

\begin{lemma}\label{lemma:ps}
The transmission success probability is 
 \begin{equation} \label{eq:p_s_unified}
        \begin{split}
            p_\mathrm{s} \!=\!\frac{q (1 - p_\mathrm{a} q)^{n-1}\!+\! p_{\mathrm{c},\mathrm{s}} \left( 1 - n q p_\mathrm{a} ( 1 - p_\mathrm{a} q)^{n-1} \right)}{p_\mathrm{T}},
        \end{split}
    \end{equation}
 where $p_\mathrm{T}$ represents the total transmission attempt probability, defined as
 \begin{equation} \label{eq:p_T}
     p_\mathrm{T} = \!q (1\!-\!p_\mathrm{a} q)^{n-1} \!+\! p_{\mathrm{ac}}(s)\left(1\!-\! n q p_\mathrm{a} ( 1 - p_\mathrm{a} q)^{n-1} \right),
 \end{equation}
 and 
 $p_{\mathrm{c},\mathrm{s}}$ can be expressed as
    \begin{equation}
p_{\mathrm{c},\mathrm{s}}=p_{\mathrm{ac}}(s)(1-p_\mathrm{a} p_{\mathrm{ac}}(s))^{n-1}.
\end{equation}
\end{lemma}
\begin{IEEEproof}
Please see Appendix~\ref{app:proof_A}.
\end{IEEEproof}

Next, we derive the node active probability $p_\mathrm{a}$ of a typical source node, which is related to its EH state. 
We model the evolution of the harvested energy as a discrete-time Markov chain, defined over a countably infinite state space $\mathcal{S} = \{0, 1, 2, \dots\}$.
Let $m$ represent the energy level of a typical source node at the onset of a time slot.
The transition probabilities of the Markov chain are determined by the interplay between the EH process with rate $\xi$ and the specific energy-depletion procedures (which depend on the transmission protocol adopted).

We further introduce two indexing sets: $\sigma \in \{s, a\}$ that indicates the operation regimes, where $s$ stands for silent and $a$ denotes active, and $y \in \{i, h, r, e, u, d\}$, which represents the specific outcome of the energy change.
Subsequently, we denote the transition probability by $P_{\sigma, y}$ and categorize the state transitions by the following:
\begin{itemize}
    \item Silent Regime ($0 \le m \le M$): The energy storage is insufficient to support a complete status update. Thus, the node remains silent, and the state transitions are governed solely by the EH process. Two transitions are possible:
    \begin{enumerate}
        \item $m \rightarrow m+1$: This transition occurs if an energy unit arrives. It corresponds to the \textit{Harvest} outcome $h$ which happens with probability $P_{s,h} = \xi$.
        \item $m \rightarrow m$: This transition takes place if no energy unit arrives. It corresponds to the \textit{Ideal} outcome $k$ which happens with probability $P_{s,i} = 1-\xi$.
    \end{enumerate}

    \item Active Regime ($m \ge M+1$): The source node has accumulated sufficient energy for status update. 
    Notably, the state transitions depend on the \textit{access mechanism} and the energy arrival patterns. In total, six transitions are possible:
    \begin{enumerate}
        \item $m \rightarrow m+1$: The node successfully accumulates one additional energy unit. This corresponds to the \textit{Harvest} outcome $h$ with probability $P_{a,h}$. % with $\Lambda_h=+1$
        \item $m \rightarrow m$: The node energy state remains unchanged, with a net energy change of zero. This corresponds to the \textit{Ideal} outcome $i$ with probability $P_{a,i}$.   %$\Lambda_k=0$
        
        \item $m \rightarrow m-1$: The node expends energy only on probing without an update. This corresponds to the \textit{Reservation} outcome $r$ with probability $P_{a,r}$.  %$\Lambda_r=-1$
        \item $m \rightarrow m-(M-1)$: A successful transmission coincides with an energy arrival that partially offsets the cost. This corresponds to the \textit{Economical Update} outcome $e$ with probability $P_{a,e}$.
        \item $m \rightarrow m-M$: A standard successful transmission occurs. This corresponds to the \textit{Standard Update} outcome $u$ with probability $P_{a,u}$.
        \item $m \rightarrow m-(M+1)$: A successful transmission occurs, and the probing cost is fully incurred without energy replenishment. This corresponds to the \textit{Deep Update} outcome $d$ with probability $P_{a,d}$.
    \end{enumerate}
\end{itemize}

% Depending on whether energy harvesting exceeds consumption, the system operates in the energy-constrained regime (ECR), characterized by a steady-state distribution, or in the energy-sufficient regime (ESR). Using the above notation, we summarize the steady-state behavior of the Markov chain as follows. 
The system operates in the energy-constrained regime (ECR) or energy-sufficient regime (ESR), depending on whether harvesting exceeds consumption. The steady-state behavior follows.

\begin{lemma} \label{lemma:S}
    \textit{The Markov chain has its steady-state distribution if the following condition holds: 
\begin{equation} \label{eq:esr_not_condition}
     (M+1)P_{a,d} + M P_{a,u} + (M-1)P_{a,e} + P_{a,r} > P_{a,h},
\end{equation} 
where the steady-state probability $S_m$ is given by
\begin{small}
\begin{equation} \label{eq:distribution}
    S_m\!=\! 
    \begin{cases}
        S_0,  &m=0,\\ \\
        \begin{aligned}
             & S_0 \left( 1 + z + \frac{P_{a,u}}{P_{a,d}} \right) \\
             & + S_0 \left( \tfrac{P_{a,e} + z P_{a,u} + z^2 P_{a,d}}{P_{a,d}} \right) \left( \tfrac{1 - z^{m-1}}{1-z} \right), 
        \end{aligned} &m\in[1,M-1],\\ \\
        S_0  \frac{P_{a,h}}{z P_{a,d}}, ~~&m=M,\\ \\
        S_0 \left( \frac{P_{s,h}}{P_{a,d}} \right) z^{m-M-1}, \qquad &m \geq M+1,
    \end{cases}
\end{equation}
\end{small}with the expression of $S_0$ given by \eqref{eq:S0} at the top of next page, and the variable $z \in (0,1)$ is the root of 
\begin{equation} \label{eq:z_general:main}
\begin{split}
   P_{a,d}z^{M+2} &+ P_{a,u} z^{M+1} + P_{a,e} z^{M} + P_{a,r}z^2 \\
   &\quad\quad\quad\quad+(P_{a,i}-1)z + P_{a,h}= 0.
\end{split}
\end{equation}
}
\end{lemma}
\begin{IEEEproof}
Please see Appendix~\ref{app:proof_B}.
\end{IEEEproof}

\begin{figure*}[t]
%\hrule
\begin{align}\label{eq:S0}
        S_0\!=\!\Biggl( 1 \!+\! \frac{P_{a,h}}{z P_{a,d}}\!+\! \frac{(M\!-\!1)(P_{a,e} + P_{a,u} + P_{a,d}) + P_{s,h}}{P_{a,d}(1-z)} - \frac{P_{a,e} + z P_{a,u} + z^2 P_{a,d}}{P_{a,d}(1-z)^2} - \frac{z^2 P_{a,r} + z(P_{a,i}-1) + P_{a,h}}{z P_{a,d}(1-z)^2} \Biggr)^{-1},
\end{align}
 \hrule
\end{figure*}

Consequently, we derive the active probability of a typical node as follows: 
\begin{equation} \label{eq:p_a_unified}
    \begin{split}
        p_\mathrm{a}
        = \min\left\{\sum_{m=M+1}^{\infty} \!\! S_m,1\right\} = \min\left\{\frac{S_0 \xi}{P_{a,d}(1-z)},1\right\}.
    \end{split}
\end{equation}

By substituting \eqref{eq:p_a_unified} into \eqref{eq:p_s_unified}, we can obtain the complete expression for the transmission success probability. 
Next, we characterize $\mathbb{E}[T]$ and $\mathbb{E}[T^2]$.
Note that the transmission interval $T$ between two consecutive successful transmissions is composed of the access waiting time $T_A$ and the energy accumulation time $T_E$, i.e., $T = T_A + T_E$. The statistical properties of these components are summarized in the following lemma.

\begin{lemma} \label{lemma:time_moments}
The first and second moments of $T_A$ are 
\begin{equation} \label{eq:TA_moments}
    \mathbb{E}[T_A]= \frac{1}{p_\mathrm{T}}, \quad \mathbb{E}[T_A^2]= \frac{2-p_\mathrm{T}}{p_\mathrm{T}^2},
\end{equation}
and the first and second moments of $T_E$ are 
\begin{equation} \label{eq:E_TE_unified}
    \mathbb{E}[T_E]=\sum_{l=1}^{M}\mathbb{P}(Q=l) \frac{l}{\xi},
\end{equation} 
and
\begin{equation} \label{eq:E_TE2_unified}
    \mathbb{E}[T_E^2]=\sum_{l=1}^{M}\mathbb{P}(Q=l) \frac{l(l-\xi+1)}{\xi^2},
\end{equation}
where $\mathbb{P}(Q=l)$ represents the probability distribution of the energy deficit $Q$ after a transmission attempt, defined as:
%. Based on the consumption weight coefficients $\Omega$, $\mathbb{P}(Q=l)$ 
\begin{equation} \label{eq:prob_Q}
\begin{split}
   &\mathbb{P}(Q=l)=\\
        &\begin{cases}
            \frac{\Omega_{deep} S_{2M+2-l} + \Omega_{std} S_{2M+1-l} + \Omega_{eco} S_{2M-l}}{p_\mathrm{a}}, & 1 \le l\!\le \!M\!-\!1, \\
            \frac{\Omega_{deep} S_{M+2} + \Omega_{std} S_{M+1}}{p_\mathrm{a}}, & l = M, \\
            \frac{\Omega_{deep} S_{M+1}}{p_\mathrm{a}}, & l = M+1,
        \end{cases}    
\end{split}
\end{equation}
where $\Omega_{deep}$, $\Omega_{std}$, and $\Omega_{eco}$ denote the conditional probabilities that the previous transmission round consumed $M+1$, $M$, and $M-1$ energy units, respectively.
\end{lemma}
\begin{IEEEproof}
%The proof is straightforward and is omitted for brevity.
Please see Appendix~\ref{app:proof_C}.
\end{IEEEproof}

Consequently, we have $\mathbb{E}[T] = \mathbb{E}[T_A] + \mathbb{E}[T_E]$ and $\mathbb{E}[T^2] = \mathbb{E}[T_A^2] + 2\mathbb{E}[T_A]\mathbb{E}[T_E] + \mathbb{E}[T_E^2]$.
Toward this end, the network-average AoI can be obtained as follows.
\begin{theorem}
    If condition \eqref{eq:esr_not_condition} holds, the network-average AoI is
    \begin{equation} \label{eq:unified aoi}
        \begin{split}
            \bar{\Delta} 
            &=\frac{1}{p_\mathrm{T} p_\mathrm{s}} +\left(\frac{1}{p_\mathrm{s}}-1\right)\mathbb{E}[T_E]+  \frac{\mathbb{E}[T_E^2]+\mathbb{E}[T_E]}{2\left(\frac{1}{p_\mathrm{T}}+\mathbb{E}[T_E]\right)},
        \end{split}
    \end{equation}
    where $p_\mathrm{s}$, $p_\mathrm{T}$, $p_\mathrm{a}$ and $\mathbb{E}[T_E]$ are given in \eqref{eq:p_s_unified}, \eqref{eq:p_T},   \eqref{eq:p_a_unified} and \eqref{eq:E_TE_unified}, respectively.    
    If condition \eqref{eq:esr_not_condition} is not satisfied, the network-average AoI is
    \begin{equation} \label{eq:unified_aoi_esr}
        \begin{split}
            \bar{\Delta} = \frac{1}{q (1 \!- \!q)^{n-1}\!+\! p_{\mathrm{ac}}(s)(1\!-\! p_{\mathrm{ac}}(s))^{n-1} \!\left( 1\!-\!nq(1\!-\!q)^{n-1} \right)},
        \end{split}
    \end{equation}
    where $p_{\mathrm{ac}}(s)$ is given by \eqref{eq:p_access_cases}.
\end{theorem}

Notably, \eqref{eq:unified aoi} characterizes the network-average AoI in ECR, while \eqref{eq:unified_aoi_esr} provides the network-average AoI in ESR.

In what follows, we leverage the result above to examine AoI under the three different channel access mechanisms. 

\subsection{Case Studies}
Since the derivations for the three cases follow a similar procedure, we present the detailed derivation only for the most involved AUC case, while omitting the derivations for RUC and SAFC for brevity.

\subsubsection{AUC} 
Upon a channel reservation failure, all active source nodes can contend for the channel access during the subsequent data transmission phase.
As such, the harvested energy states admit the following transition probabilities 
\begin{equation}
    \begin{cases}
    P_{a,d} = (1-\xi)q \left( P_{0}^{n-1} + (1-P_{0}^{n-1})\eta \right), \\
     P_{a,u} = \xi q \left( P_{0}^{n-1} + (1-P_{0}^{n-1})\eta \right)  \\
    ~~~~~~+ (1-\xi)(1-q)(1-P_{1}^{n-1})\eta,\\
    P_{a,e} = \xi(1-q)(1-P_{1}^{n-1})\eta, \\
     P_{a,r} = (1-\xi)q (1-P_{0}^{n-1})(1-\eta), \\
      P_{a,i} = \xi q (1-P_{0}^{n-1})(1-\eta) \\
    ~~~+ (1-\xi)(1-q) \left( P_{1}^{n-1} \!+\!(1\!-\!P_{1}^{n-1})(1-\eta) \right), \\
     P_{a,h}= \xi(1-q) \left( P_{1}^{n-1} + (1-P_{1}^{n-1})(1-\eta) \right) 
    \end{cases}
\end{equation}
where $P_{0}^{n-1}$ and $P_{1}^{n-1}$ denote the probabilities that zero and exactly one other node succeed in reservation, respectively:
\begin{equation} \label{eq:p_res_n-1}
    \begin{split}
        P_{0}^{n-1} &= (1 - p_\mathrm{a} q)^{n-1}, \\
        P_{1}^{n-1} &= (n-1) p_\mathrm{a} q( 1 - p_\mathrm{a} q)^{n-2},
    \end{split}
\end{equation}
where $p_\mathrm{a}$ stands for the node active probability under the AUC mechanism. 

Consequently, we can obtain the active probability of a typical source node as
\begin{equation} \label{eq:p_a_AUC}
    p_\mathrm{a} = \tfrac{\xi}{ M \Bigl( q \bigl( P_{0}^{n-1} + (1-P_{0}^{n-1})\eta \bigr)
    + \eta\left(1-q\right)\left(1-P_{1}^{n-1}\right) \Bigr) + q}.
\end{equation}
Then, the transmission attempt probability and transmission success probability are respectively given as
\begin{equation} \label{eq:pT_case1}
    \begin{split}
        p_\mathrm{T} 
        &= \eta + q(1 - n \eta p_\mathrm{a})(1 - p_\mathrm{a} q)^{n-1},
    \end{split}
\end{equation}
and 
\begin{equation} \label{eq:ps_case1}
    \begin{split}
        p_\mathrm{s}\!=\!\tfrac{ \eta (1 - p_\mathrm{a} \eta)^{n-1}+q (1- p_\mathrm{a} q)^{n-1}\Big(1-n\eta p_\mathrm{a} (1-p_\mathrm{a} \eta)^{n-1} \Big)}{\eta + q\left( 1 - n\eta p_\mathrm{a} \right) (1 - p_\mathrm{a} q)^{n-1} }.
    \end{split}
\end{equation}

Moreover, the weight coefficients $\Omega$ in AUC are:  $\Omega_{deep}= (1-\xi)q$, $\Omega_{std} = \xi q+(1-\xi)(1-q)$, and $\Omega_{eco}= \xi(1-q)$.
Then, following Lemma~\ref{lemma:time_moments}, we have 
\begin{equation} \label{eq:ETE_case1}
    \begin{split}
        \mathbb{E}[T_E] 
        &= \frac{1}{\xi} \left( M+q-\xi - \frac{z}{1-z} \right),
    \end{split}
\end{equation}
and
\begin{equation} \label{eq:ETE2_case1}
    \begin{split}
        \mathbb{E}[T_E^2] 
        &= \frac{1}{\xi^2} \Bigg( \left( M + q - \xi - \frac{z}{1-z} \right) \\
        & \times \left( M + q + 1 - 2\xi - \frac{z}{1-z} \right) + \frac{z}{(1-z)^2} \Bigg).
    \end{split}
\end{equation}

Finally, by substituting \eqref{eq:pT_case1}, \eqref{eq:ps_case1}, \eqref{eq:ETE_case1}, and \eqref{eq:ETE2_case1} into \eqref{eq:unified aoi}, we obtain the closed-form expression for the network-average AoI, where the result is summarized as follows{\footnote{While the complete expression of the network-average AoI can be derived, it is highly complicated; hence, we only present an approximation. }}. 

\begin{corollary}
    In the AUC scheme, if the following holds
    \begin{small}
    \begin{equation} \label{eq:esr_not_condition1}
        \begin{split}
            \xi\!<\!q\!+\!M\bigg(\!\eta\!+\!q(\!1\!-\!p_\mathrm{a} q)^{n\!-\!2} \!\Big(\!(\!1\!-\!\eta)(\!1\!-\!p_\mathrm{a})\!+\!p_\mathrm{a}(\!1\!-\!q)(1\!-\!n\eta)\!\Big)\!\!\bigg),
        \end{split}
    \end{equation}
    \end{small}
%     the network-average AoI is given by \eqref{eq:AoI:final:case1} at the top of this page, 
%     \begin{figure*}
%     \begin{equation} \label{eq:AoI:final:case1}
%     \begin{split} 
%     \bar{\Delta} 
%     &= \frac{ 1 + \frac{\eta}{\xi} \left( 1 - (1 - p_a \eta)^{n-1} \right) \left( 1 - n p_a q (1 - p_a q)^{n-1} \right)\left( M+q-\xi - \frac{z}{1-z} \right)}{q(1-p_a q)^{n-1}+ (\eta (1 - p_a \eta)^{n-1})(1 - n p_a q(1-p_a q)^{n-1}))} \\
%     &\quad + \frac{(\eta + (q(1-p_a q)^{n-1})(1 - n\eta  p_a))\bigg(\left( M+q-\xi - \frac{z}{1-z} \right)\left( M+q-\xi+\frac{1-2z}{1-z} \right)+\frac{z}{(1-z)^2} \bigg)}{2\xi\left(1 + \left(\eta + (q(1-p_a q)^{n-1}\right)\left(1 - n \eta  p_a\right)\right)\left( M+q-\xi - \frac{z}{1-z} \right)},
%     \end{split} 
%     \end{equation}
%     \hrule
%    % \vspace{-0.3cm}
%     \end{figure*}
%     in which $z$ is the root of
%     \begin{equation} \label{eq:z_general:case1}
% \begin{split}
%    P_{a,d}z^{M+2} &+ P_{a,u} z^{M+1} + P_{a,e} z^{M} + P_{a,r}z^2 \\
%    &~~~~~~~+(P_{a,i}-1)z + P_{a,h}= 0.
% \end{split}
% \end{equation}
    the network-average AoI can be tightly approximated by
    \begin{equation}
        \begin{split}
            &\bar{\Delta} \approx \frac{1 + \frac{\eta}{\xi}(1 - e^{-n p_\mathrm{a} \eta})(1 - n p_\mathrm{a} q e^{-n p_\mathrm{a} q}) C}{q e^{-n p_\mathrm{a} q} + \eta e^{-n p_\mathrm{a} \eta} (1 - n p_\mathrm{a} q e^{-n p_\mathrm{a} q} )} \\
            & + \frac{\eta + q e^{-n p_\mathrm{a} q}}{2\xi C (1 + \eta + q e^{-n p_\mathrm{a} q})}  \times \left(\! C \left(\! C \!+\! \frac{1}{ 1 \! - \! z } \right) \!+\! \frac{z}{(1\!-\!z)^2} \right),
        \end{split}
    \end{equation}
    where $C$ is
    % \begin{equation}  \label{eq:A}   
    %     A = e^{-n p_a q},
    % \end{equation}
    % \begin{equation} \label{eq:B}
    %     B = e^{-n p_a \eta},
    % \end{equation}
    \begin{equation} \label{eq:C}
        C = M + q - \xi - \frac{z}{1-z},
    \end{equation}
    and $z$ is the root of the following
    \begin{equation} \label{eq:z_general:case1}
        \begin{split}
           P_{a,d}z^{M+2} &+ P_{a,u} z^{M+1} + P_{a,e} z^{M} + P_{a,r}z^2 \\
           &~~~~~~~+(P_{a,i}-1)z + P_{a,h}= 0.
        \end{split}
    \end{equation}
If \eqref{eq:esr_not_condition1} is not satisfied, the network-average AoI reduces to
\begin{equation}
    \begin{split}
        \bar{\Delta} 
        % &= \frac{1}{A + B(1 - n p_a A)} \\
        = \frac{1}{q(1-q)^{n-1}+ \eta(1-\eta)^{n-1}(1 - n q(1-q)^{n-1})}.
        % &= \frac{1}{p_\mathrm{T} p_\mathrm{s}}
    \end{split}
\end{equation}
\end{corollary}

\subsubsection{RUC}
This mechanism imposes stricter restrictions on the data transmission phase, where only the probing nodes are allowed to transmit after a reservation failure.  The network-time average AoI expression under the RUC scheme is as follows. 
\begin{corollary}
    In RUC scheme, if the following holds
    \begin{equation} \label{eq:esr_not_condition2}
        %q + M q \left( \eta + (1-\eta)(1-q)^{n-1} \right) > \xi,
        q \left( 1 + M \left( \eta + (1-\eta)(1-p_\mathrm{a} q)^{n-1} \right) \right) > \xi,
    \end{equation}
    the network-average AoI can be tightly approximated by
\begin{equation}
    \begin{split}
        \bar{\Delta} &\approx \frac{\xi + q C \eta \left(1 - e^{-n p_\mathrm{a} q}\right)\left(1 - e^{-n q p_\mathrm{a} \eta}\right)}{\xi q \left( e^{-n p_\mathrm{a} q} + \eta e^{-n q p_\mathrm{a} \eta} \left(1 - e^{-n p_\mathrm{a} q}\right) \right)} \\
        &\quad + \frac{q \left( \eta\left(1 - e^{-n p_\mathrm{a} q}\right) + e^{-n p_\mathrm{a} q} \right)}{2\xi \left( \xi + q C \left( \eta\left(1 - e^{-n p_\mathrm{a} q}\right) + e^{-n p_\mathrm{a} q} \right) \right)} \\
        &\quad \times \left( C \left( C + \xi + \frac{4z}{1-z} \right) + \frac{3z^2-z}{(1-z)^2} \right),
    \end{split}
\end{equation}
    where $C$ is given by \eqref{eq:C} and $p_\mathrm{a}$ is given by
\begin{equation} \label{eq:p_a_RUC}
\small
     p_\mathrm{a} \!=\!\frac{\xi z (1\!-\!z^M\!) \left( (1\!-\!\xi)z \!+\! \xi \right)}{ q z (\!1\!-\!z^M\!) \!\left( \!(1\!-\!\xi)z \!+\! \xi \right) \!+\! M \!(\!1\!-\!z) \!\left( \xi(1\!-\!q)\! -\! (1\!-\!\xi)q z \right)},
\end{equation}
in which $z\in (0,1)$ is a root of
\begin{equation} \label{eq:z_RUC}
    \begin{split}
        P_{a,d}z^{M+2}
       \!+\!\!P_{a,u} z^{M+1}
       \!+\!P_{a,r}z^2\!\!+\!(P_{a,i}\!-\!1)z
       \!+\!P_{a,h}\!=\!0.
    \end{split}
\end{equation}
If \eqref{eq:esr_not_condition2} is not satisfied, the network-average AoI simplifies to
\begin{equation}
    \begin{split}
        \bar{\Delta} 
        % &= \frac{1}{q \left[ C(1 - \eta^2 n q p_a H) + \eta^2 H D \right]} \\
        &= \frac{1}{q} \bigg( (1-q)^{n-1} \left(1 - n \eta q (1-q\eta)^{n-1}\right) \\
    &\quad + \eta (1-q\eta)^{n-1} \left(1 - (1-q)^n\right) \bigg)^{-1}.
        % &= \frac{1}{p_\mathrm{T} p_\mathrm{s}}
    \end{split}
\end{equation}
\end{corollary}
%  The result shows that in the energy-sufficient regime, the system degenerates into a reservation-based SA network without energy constraints. The mandatory reservation acts as a filter that limits transmission opportunities compared to Case 1, preventing the AoI from reaching the optimal lower bound even when energy is abundant.
\subsubsection{SAFC}
In this scenario, if the channel probing phase fails, then none of the nodes are allowed to transmit in this round. The network-time average AoI expression under the SAFC scheme is as follows.
\begin{corollary}
 In SAFC scheme, if $q + M q (1 - p_\mathrm{a} q)^{n-1} > \xi$, the network-average AoI can be tightly approximated by
\begin{equation}
    \begin{split}
        \bar{\Delta} \approx \frac{1}{q e^{-n p_\mathrm{a} q}} + \frac{\xi q  e^{-n p_\mathrm{a} q} K}{2 \left( \xi + q  e^{-n p_\mathrm{a} q} C \right)},
    \end{split}
\end{equation}
where $C$ is given by \eqref{eq:C}, $K$ is given by
\begin{align}
    %\begin{split}
        K &= \frac{1}{\xi^2} \bigg( \frac{(M+1)(M+2-3\xi) - z(2M+2-3\xi) + \xi^2}{1-z} \notag\\
        &\quad + \frac{2z^2}{(1-z)^2} + \xi C \bigg),
    %\end{split}
\end{align}
and $p_\mathrm{a}$ is given by
\begin{equation} \label{eq:p_a_SAFC}
    \begin{split}
        p_\mathrm{a} 
        % &= \left\{ \frac{1-q(1-z)}{z} + \frac{q\eta}{\xi} \left[ (M-1) + (1-\xi)(1-z) - \frac{z\xi + z^2(1-\xi)}{1-z}(1-z^{M-1}) \right] \right\}^{-1} \\
        % &= \left[ \frac{1-q(1-z)}{z} + \frac{q\eta M}{\xi} - q\eta(1-z^M)\left( 1 + \frac{z}{\xi(1-z)} \right) \right]^{-1} \\
        % &= \left\{\tfrac{1-q(1-z)}{z} + q\eta \left[ \tfrac{M}{\xi} - (1-z^M)\left( 1 + \tfrac{z}{\xi(1-z)} \right) \right] \right\}^{-1}.
        &= \frac{1}{q}-\frac{1}{q}\left( \frac{(1-z) \big( q((1-\xi)z + \xi) - \xi \big)}{q z (z^M - 1)((1-\xi)z + \xi)} \right)^{\frac{1}{n-1}}.
    \end{split}
\end{equation}
If $q + M q (1 - p_\mathrm{a} q)^{n-1} \le \xi$, the network-average AoI takes the following expression
\begin{equation}\label{eq:51}
    \begin{split}
        \bar{\Delta} 
        % &= \frac{1}{p_\mathrm{T}} \\
        &= \frac{1}{q (1 - q)^{n-1}}.
    \end{split}
\end{equation}
\end{corollary}

\begin{remark}
To fairly compare with benchmark schemes (e.g., SA) under varying overheads, we map the theoretical results derived in discrete rounds to the physical time domain. Considering a probing overhead ratio $\delta$, the effective duration of a round extends to $1+\delta$. This extension introduces two scaling effects on the AoI:
\begin{enumerate}
    \item Effective Energy Arrival: The round to harvest energy becomes longer. For small $\xi$ and $\delta$, the effective energy arrival probability scales to $\xi' \approx (1+\delta)\xi$.
    \item Time Unit Renormalization: The calculated average AoI in rounds, $\bar{\Delta}(\xi')$, must be converted back to physical time units by multiplying with the round duration.
\end{enumerate}
Consequently, the network-average AoI in physical time, denoted by $\bar{\Delta}_{\mathrm{phy}}$, is obtained by
\begin{equation}
    \bar{\Delta}_{\mathrm{phy}} = (1+\delta) \cdot \bar{\Delta}\big|_{\xi'\leftarrow (1+\delta)\xi},
\end{equation}
where $\bar{\Delta}|_{\xi' \leftarrow (1+\delta)\xi}$ denotes the AoI expression from Corollaries 1--3 evaluated with the scaled energy arrival rate.
\end{remark}

% The proposed analytical model requires only $\approx 0.3$~s per calculation compared to $\approx 10$~s for simulation, enabling rapid performance evaluation without simulation overhead.

% ======================= %
%    Numerical Results
% ======================= %
\section{Numerical results}
This section validates the theoretical analysis via simulations and evaluates the network-average AoI. 
We consider a random access network with $n$ source nodes and one AP, varying $n$. 
The optimal parameters $q$ and $\eta$ are obtained via grid search. 
Each simulation runs for $10^5$ time slots and is averaged over 100 realizations. 
Unless otherwise stated, $\delta=1/20$, $M=7$, and $\xi=0.1$.

From Fig.~\ref{fig:change_n_xi0.1_q} and Fig.~\ref{fig:change_n_xi0.1_eta}, the optimal parameters exhibit distinct trends across different access mechanisms. 
For RUC and SAFC, $q^*$ rapidly decreases and approaches the SA theoretical baseline $1/n$ as $n$ increases, while $\eta^*$ also decays. 
This is because strict reservation constraints force nodes to suppress access attempts in order to avoid collisions, leading to increasingly conservative operation as the network becomes denser. 
In contrast, AUC maintains relatively higher levels of $q^*$ and $\eta^*$ over a wide range of $n$, since its fallback mechanism allows nodes to continue contention after reservation failure, thereby sustaining channel activity. 

These parameter trends are consistent with the AoI performance in Fig.~\ref{fig:change_n_xi0.1} and the energy dynamics in Fig.~\ref{fig:energy_consumption_rate}. 
As the network transitions from ECR to ESR, RUC and SAFC become increasingly constrained by reservation policies, which suppress transmission opportunities and prolong the energy accumulation process. 
This leads to inefficient energy utilization, where harvested energy cannot be effectively converted into successful updates. 
In contrast, AUC alleviates this limitation by maintaining more active contention, which amortizes probing overhead and better matches energy consumption with the arrival rate $\xi$, ultimately improving AoI performance.
% As a result, conservative schemes tend to accumulate energy without effectively converting it into transmissions. 
% The suppression of access attempts prolongs the expected energy collection time $\mathbb{E}[T_E]$ and reduces channel utilization, whereas AUC mitigates this effect by amortizing probing overhead and sustaining channel activity. 
% By aligning the energy consumption rate with the arrival rate $\xi$, AUC enables more efficient use of harvested energy.

Fig.~\ref{fig:change_ratio} plots the network-average AoI versus the probing-to-data slot ratio $\delta$. 
Under low energy arrival ($\xi=0.1$), AUC maintains its advantage even at high overhead levels ($\delta \geq 0.5$). 
This highlights a key trade-off: higher signaling overhead can be beneficial under energy scarcity, as increased probing reduces collisions and avoids costly energy outages, ultimately improving AoI.

\begin{figure*}[!t]
    \centering
    \subfloat[$\bar{\Delta}$ vs $n$.]{
        \includegraphics[width=0.33\linewidth]{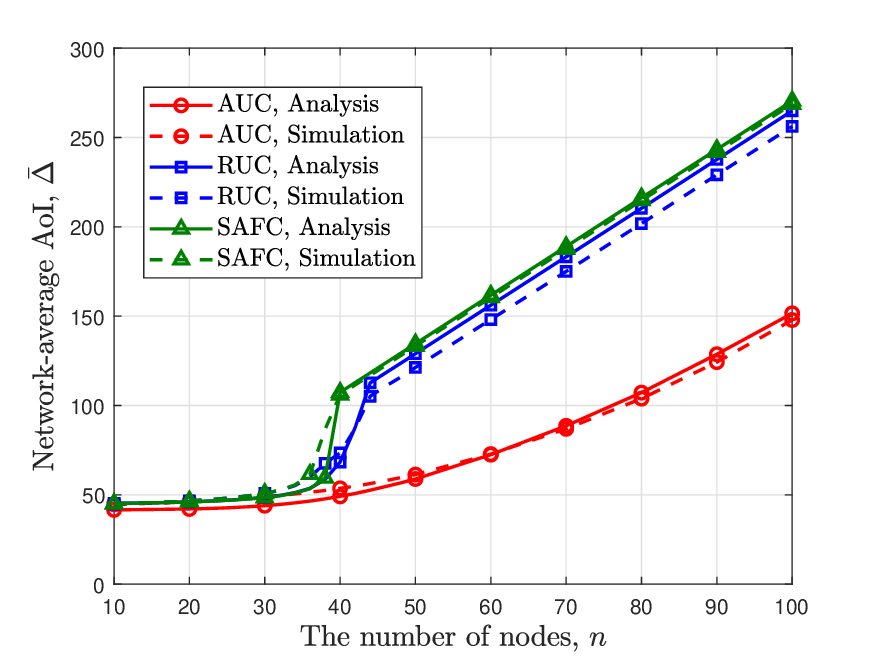}
        \label{fig:change_n_xi0.1_aoi}
    }
    \subfloat[$q$ vs $n$.]{
        \includegraphics[width=0.33\linewidth]{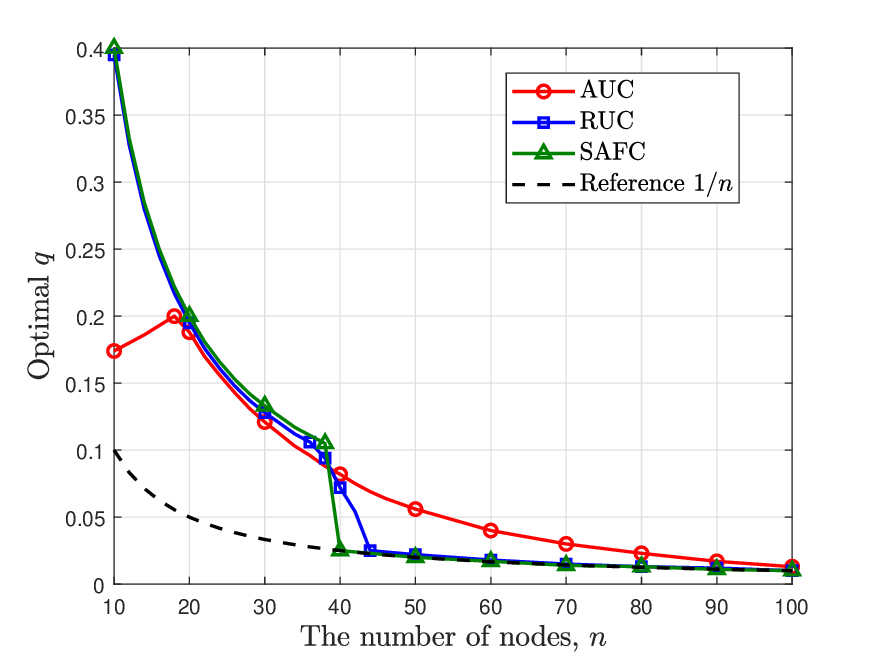}
        \label{fig:change_n_xi0.1_q}
    }
    \subfloat[$\eta$ vs $n$.]{
        \includegraphics[width=0.33\linewidth]{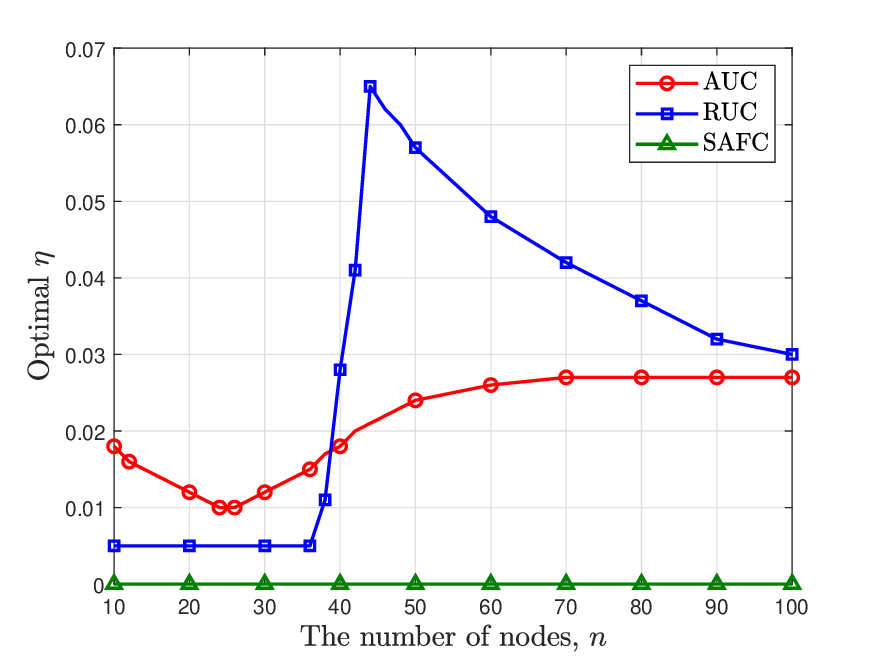}
        \label{fig:change_n_xi0.1_eta}
    }
    \caption{Comparison of network-average AoI under optimal parameters for three situations.} %the rounds per simulation run is $10^4$, the number of simulation runs is $10$.}
    \label{fig:change_n_xi0.1}
    \vspace{-0.5cm}
\end{figure*}

% \begin{figure}[!t]
%     \centering
%     \includegraphics[width=0.85\linewidth]{figure/average_energy_queue_length_xi0.1.eps}
%     \caption{Illustration of resource hoarding via average energy queue accumulation.}
%     \label{fig:average_energy_queue_length}
%     \vspace{-0.5cm}
% \end{figure}
\begin{figure}[!t]
    \centering
    \includegraphics[width=0.8\linewidth]{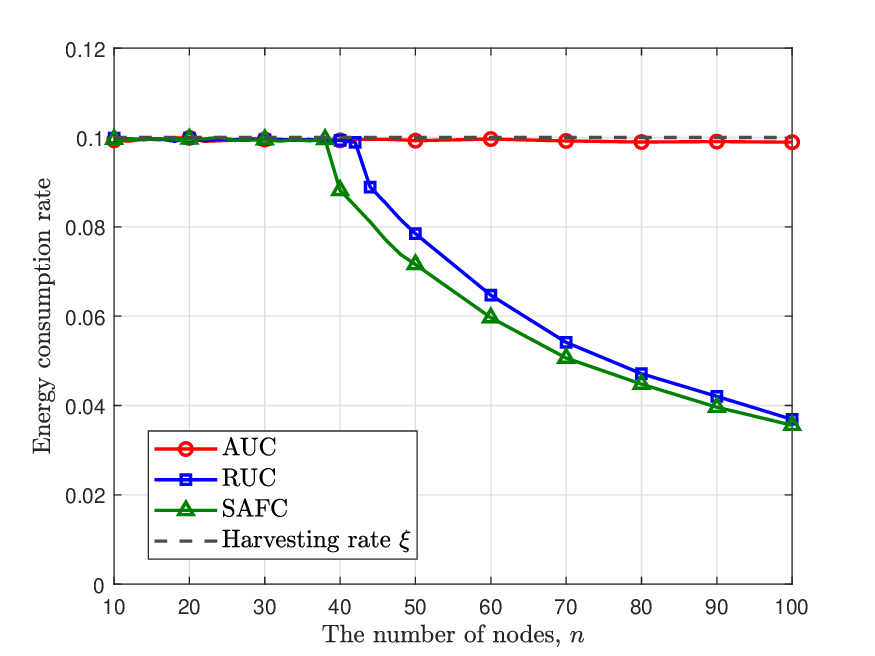}
    \caption{Demonstration of utilization efficiency via energy consumption rate.}
    \label{fig:energy_consumption_rate}
    \vspace{-0.5cm}
\end{figure}

\begin{figure}[!t]
    \centering
    \includegraphics[width=0.8\linewidth]{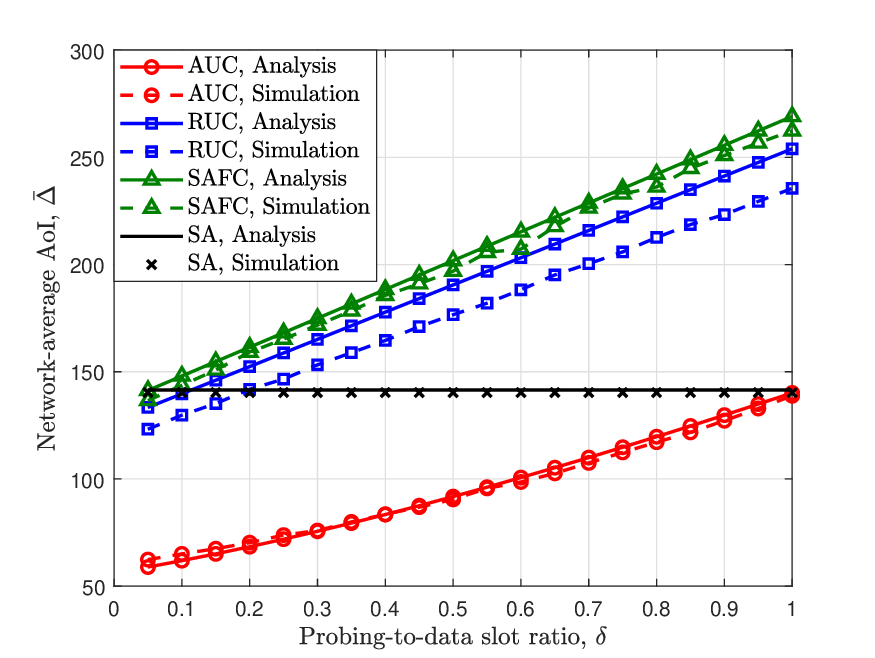}
    \caption{Comparison of simulated and theoretical network-average AoI under optimal parameters. $n=50$.}
    \label{fig:change_ratio}
    \vspace{-0.7cm}
\end{figure}
\section{Conclusion}
This paper analyzed the network-average AoI in EH-enabled reservation networks. The results show that handling reservation failures is decisive for performance, and AUC achieves the best freshness in the energy-constrained regime because avoiding costly data collisions outweighs the probing overhead. We further show that the superiority of reservation over direct transmission is contingent on energy availability. As the energy arrival rate decreases, the penalty of energy wastage caused by data collisions dominates the signaling cost of probing, making reservation-based access the more robust choice for maintaining information freshness.

\bibliography{ref}
\appendices
\section{} \label{app:proof_A}
Based on conditional probability, $p_\mathrm{s}$ is expressed as
\begin{equation}
p_\mathrm{s}=\frac{P_\mathrm{res,s}+P_\mathrm{res,f}^\mathrm{network}p_{\mathrm{c},\mathrm{s}}}{p_\mathrm{T}},
\end{equation}
where $p_{\mathrm{c},\mathrm{s}}$ denotes the contention success probability. The probability of a specific node successfully reserving a slot is
\begin{equation}
    P_\mathrm{res,s}^\mathrm{node} = q p_\mathrm{a} (1 - p_\mathrm{a} q)^{n-1}.
\end{equation}
The network-wide reservation failure probability, representing the event where no node succeeds, is 
\begin{equation}
    P_\mathrm{res,f}^\mathrm{network} = 1 - nP_\mathrm{res,s}^\mathrm{node} = 1 - nqp_\mathrm{a}(1-p_\mathrm{a} q)^{n-1}. 
\end{equation}
Substituting these terms into \eqref{eq:p_s_unified} concludes the proof.

\section{} \label{app:proof_B}
We derive the steady-state probabilities $S_m$ via the global balance equations. 
For the state $m=0$, we have
\begin{equation} \label{eq:balance_0}
    (1-P_{s,i}) S_0 = P_{a,d} S_{M+1}.
\end{equation}
For the state $m=1$, we have
\begin{equation} \label{eq:balance_1}
\begin{split}
    P_{s,h} S_1 &= P_{s,h} S_0 + P_{a,u} S_{M+1} + P_{a,d} S_{M+2}.
\end{split}
\end{equation}
For the state $2 \le m \le M-1$, we have
\begin{equation} \label{eq:balance_2_M-1}
\begin{split}
    P_{s,h} S_m &= P_{s,h} S_{m-1} + P_{a,e} S_{m+M-1}\\
    &\quad+ P_{a,u} S_{m+M} + P_{a,d} S_{m+M+1}.
\end{split}
\end{equation}
For the state $m=M$, we have
\begin{equation} \label{eq:balance_M}
\begin{split}
    P_{s,h} S_M &= P_{s,h} S_{M-1} + P_{a,e} S_{2M-1}\\
    &\quad + P_{a,u} S_{2M} + P_{a,d} S_{2M+1}.
\end{split}
\end{equation}
For the state $m=M+1$, we have
\begin{equation} \label{eq:balance_M+1}
\begin{split}
    (1 - P_{a,i}) S_{M+1} &= P_{s,h} S_M + P_{a,r} S_{M+2} + P_{a,e} S_{2M} \\
    &\quad + P_{a,u} S_{2M+1} + P_{a,d} S_{2M+2}.
\end{split}
\end{equation}
For the state $m \ge M+2$, we have
\begin{equation} \label{eq:balance_active_internal}
\begin{split}
    (1 - P_{a,i}) S_m &= P_{a,h} S_{m-1} + P_{a,r} S_{m+1} \\
    &\hspace{-4em}+ P_{a,e} S_{m+M-1}+ P_{a,u} S_{m+M} + P_{a,d} S_{m+M+1}.
\end{split}
\end{equation}
Based on the equation \eqref{eq:balance_active_internal}, the steady state probability $S_m$, $m\geq M+1$ can be expressed as $S_m=Cz^m, m\geq M+1$,
% \begin{equation}
%     S_m=Cz^m, \qquad m\geq M+1,
% \end{equation}
where $z \in (0,1)$ is the root of 
\begin{equation} \label{eq:z_general}
\begin{split}
   P_{a,d}z^{M+2} &+ P_{a,u} z^{M+1} + P_{a,e} z^{M} + P_{a,r}z^2 \\
   &+ (P_{a,i}-1)z + P_{a,h}= 0.
\end{split}
\end{equation}
Using the equation \eqref{eq:balance_0}, we have
\begin{equation} \label{eq:10}
   S_{M+1}=\tfrac{P_{s,h}}{P_{a,d}}S_0.
\end{equation}
Therefore, for $m\geq M+1$, we can express $S_m$ as
\begin{equation}
    S_m = S_{M+1} z^{m-(M+1)} = \tfrac{P_{s,h}}{P_{a,d}} z^{m-M-1} S_0.
\end{equation}
Using the equation \eqref{eq:balance_1}, we have
\begin{equation}
\begin{split}
     S_1 
     %&= S_0+\frac{P_{a,u}}{P_{s,h}}S_{M+1}+\frac{P_{a,d}}{P_{s,h}}S_{M+2}\\
     %&= S_0+\frac{P_{a,u}}{P_{s,h}} \left( \frac{1-P_{s,h}}{P_{a,d}} S_0 \right) +\frac{P_{a,d}}{P_{s,h}} \left( \frac{1-P_{s,h}}{P_{a,d}} z S_0 \right) \\
     &= S_0 \left( 1 + z + \tfrac{P_{a,u}}{P_{a,d}}\right).
\end{split}
\end{equation}
Using the equation \eqref{eq:balance_2_M-1}, for $2\leq m\leq M-1$, we have
\begin{equation}
\begin{split}
     S_m - S_{m-1} 
     %&= \frac{1}{P_{s,h}} \Big( P_{a,e} S_{m+M-1} + P_{a,u} S_{m+M} \\
     %&\quad + P_{a,d} S_{m+M+1} \Big) \\
     %&= \frac{S_{M+1}}{P_{s,h}} \left( P_{a,e} z^{i-2} + P_{a,u} z^{i-1} + P_{a,d} z^{i} \right) \\
     &= S_0 \left( \tfrac{P_{a,e} + z P_{a,u} + z^2 P_{a,d}}{P_{a,d}}\right)z^{m-2}.
\end{split}
\end{equation}
Summing this recursion leads to
\begin{equation}
\begin{split}
    S_m \!=\!S_1\!+\!S_0 \left( \tfrac{P_{a,e}\!+\!z P_{a,u}\!+\!z^2 P_{a,d}}{P_{a,d}} \right)\!\!\left( \tfrac{1-z^{m-1}}{1-z} \right).
\end{split}
\end{equation}
Substituting the characteristic equation \eqref{eq:10} into \eqref{eq:balance_M+1} yields
\begin{equation}
    \begin{split}
         S_M 
         %&= \frac{1}{P_{s,h}} \Big( (1 - P_{a,i})S_{M+1} - (P_{a,r}S_{M+2} + P_{a,e}S_{2M} \notag \\
         %&\quad + P_{a,u}S_{2M+1} + P_{a,d}S_{2M+2}) \Big) \notag \\
         %&= \frac{S_{M+1}}{P_{s,h}} \frac{P_{a,h}}{z} \notag \\
         &= S_0 \left( \tfrac{P_{a,h}}{z P_{a,d}} \right).
    \end{split}
\end{equation}
The existence of the steady-state distribution requires a root $z \in (0,1)$ for \eqref{eq:z_general}. Therefore, we define the characteristic function $f(z)$ as
\begin{equation}
    \begin{split}
        f(z) &= P_{a,d}z^{M+2} + P_{a,u} z^{M+1} + P_{a,e} z^{M} \\
        &\quad + P_{a,r}z^2 + (P_{a,i}-1)z + P_{a,h}.
    \end{split}
\end{equation}
Observing $f(1) = 0$ and $f(0) = P_{a,h} > 0$, we note that $f(z)$ is strictly convex for $z > 0$ since $f''(z) > 0$. Thus, a unique root exists in $(0,1)$ if and only if $f'(1) > 0$. This yields the condition \eqref{eq:esr_not_condition}.

\section{} \label{app:proof_C}
We derive the moments of the transmission interval by separately characterizing the access waiting time $T_A$ and the energy accumulation time $T_E$.

We begin with $T_A$. Once a node is active, each communication round leads to a transmission attempt with probability $p_\mathrm{T}$, independently across rounds. Therefore, $T_A$ follows a geometric distribution with parameter $p_\mathrm{T}$, i.e.,
\begin{equation}
    \mathbb{P}(T_A=t)=(1-p_\mathrm{T})^{t-1}p_\mathrm{T}, \qquad t=1,2,\dots
\end{equation}
and its first two moments are given by \eqref{eq:TA_moments}.

We next turn to $T_E$. Let $Q$ denote the energy deficit immediately after a transmission attempt, namely, the number of harvested energy units required for the node to become active again. Conditioned on $Q=l$, the node must accumulate $l$ energy units through the Bernoulli EH process with rate $\xi$. Following the same argument as in \cite{zhao2025age}, $T_E|Q=l$ follows a negative binomial distribution:
\begin{equation}
    \mathbb{P}(T_E=t|Q=l)=\binom{t-1}{l-1}\xi^l(1-\xi)^{t-l}, \qquad t=l,l+1,\dots
\end{equation}
Hence, the corresponding conditional first and second moments are
\begin{equation}\label{eq:T_E|Q=l}
    \mathbb{E}[T_E|Q=l]=\frac{l}{\xi},
\end{equation}
and
\begin{equation}\label{eq:T_E^2|Q=l}
    \mathbb{E}[T_E^2|Q=l]=\frac{l(l-\xi+1)}{\xi^2}.
\end{equation}

Applying the law of total expectation to \eqref{eq:T_E|Q=l} and \eqref{eq:T_E^2|Q=l} gives \eqref{eq:E_TE_unified} and \eqref{eq:E_TE2_unified}, once the distribution of $Q$ is specified. We therefore characterize $\mathbb{P}(Q=l)$ as follows.

Consider a node in state $m\ge M+1$ at the beginning of a transmission round. Depending on whether the transmission results in a deep, standard, or economical update, the post-transmission energy state becomes $m-(M+1)$, $m-M$, and $m-(M-1)$, respectively. Since the node becomes active again once its energy level reaches $M+1$, for a given deficit level $l$, the corresponding pre-transmission state satisfies
\begin{equation}
    m=
    \begin{cases}
        2M+2-l, & \text{deep update},\\
        2M+1-l, & \text{standard update},\\
        2M-l,   & \text{economical update}.
    \end{cases}
\end{equation}
Accordingly, all three cases are feasible for $1\le l\le M-1$; for $l=M$, only the first two remain feasible; and for $l=M+1$, only the deep-update case is possible.

Conditioning on the event that the node is active at the beginning of the round, which occurs with probability $p_\mathrm{a}=\sum_{m=M+1}^{\infty}S_m$, $\mathbb{P}(Q=l)$ is obtained by summing the probabilities of all feasible pre-transmission states leading to deficit $l$, weighted by the corresponding conditional outcome probabilities $\Omega_{deep}$, $\Omega_{std}$, and $\Omega_{eco}$. Therefore, $\mathbb{P}(Q=l)$ is given by \eqref{eq:prob_Q}. Substituting \eqref{eq:prob_Q} into \eqref{eq:E_TE_unified} and \eqref{eq:E_TE2_unified} completes the proof.

\end{document}